\documentclass[nolinenumbers,undated,a4paper,envcountsame,envcountsech]{lnthi}
\def\texorpdfstring#1#2{#1}
\setpagewiselinenumbers
\usepackage[USenglish]{babel}

\usepackage{graphicx}

\title{Space-Efficient Biconnected Components and Recognition of
Outerplanar Graphs}

\author{Frank Kammer\inst{1}, Dieter Kratsch\inst{2}, and Moritz Laudahn\inst{1}}

\institute{Institut f\"ur Informatik, Universit\"at Augsburg, 86135
Augsburg, Germany
\email{\{kammer,moritz.laudahn\}@informatik.uni-augsburg.de}
\and
LITA, Universit\'e de Lorraine, Metz, France\\
\email{dieter.kratsch@univ-lorraine.fr}}

\begin{document}

\maketitle{}

\begin{abstract}
     We present space-efficient algorithms for computing cut vertices in
   a given graph with $n$ vertices and $m$ edges in
linear
 time using $O(n+\min\{m,n\log \log n\})$ bits. 
With the same time and using $O(n+m)$ bits, we can compute the 
  biconnected components of a graph. 
  We use this result to show an algorithm for the recognition of
(maximal) outerplanar graphs in $O(n\log \log n)$ time using $O(n)$
 bits.\vspace{1mm}


\end{abstract}

\pagestyle{plain}
\thispagestyle{plain}

\section{Introduction}

Nowadays the use of small mobile devices like tablets and smartphones is ubiquitous. 
Typically they will not be equipped with large memory and common actions like storing
(many) pictures may even decrease the available memory significantly. 
This triggers the interest in data structures 
and algorithms being space-efficient. 
(Time) Efficient algorithms are a classical subject in computer science. 
The corresponding 
algorithms course is often based on the textbook of Cormen et al.~\cite{CorLRS09}.
There is also a long tradition in the development of algorithms that use as few bits as
possible;
famous results are the two results of Savitch~\cite{Sav70} and Reingold~\cite{Rei08} on reachability in directed
and undirected graphs, respectively. However, the running times of their
algorithms are far away from the fastest algorithms for that problem and are
therefore of small practical interest.
Moreover, Edmonds et al.~\cite{EdmPA99} have shown in the so-called
NNJAG model that only a slightly
sublinear working-space bound is possible for an algorithm that solves the reachability problem when required
to run in polynomial time.
This motivates the recent interest in
{\em space-efficient algorithms}, 
i.e., algorithms that use 
as few working space as possible under the condition that their running time (almost) matches
the running time of the best algorithm(s) without any space restrictions.

A useful model of computation for the development of algorithms in that
context is the word RAM with a read-only input, a read-write working memory,
and a write-only output. As usual, we assume that for a given instances of size
$n$, the word size is $\Omega(\log n)$.
One of the first problems considered for space-efficient algorithms is
sorting~\cite{MunP80}, which was finally solved to
optimality~\cite{Bea91,PagR98}.
Other researchers considered problems in geometry~\cite{AsaBBKMRS14,AsaMRW11,BarKLSS13}.
There are only few papers with space-efficient graph algorithms.
Asano et al.~\cite{AsaIKKOOSTU14} focused on depth-first search and Elmasry et
al.~\cite{ElmHK15} considered
depth-first search, breadth-first search, %
(strongly) connected
 components, topological sorting and shortest path. Moreover, Datta et
 al.~\cite{DatKM16}
 gave an space-efficient matching algorithm for sparse graphs.

We continue this work on space-efficient
      graph algorithms 
and consider the basic problems to compute the cut vertices and to 
decompose a given undirected graph into its
biconnected components. Tarjan's linear time algorithm~\cite{Tar72} solving
this problem uses a DFS and 
has
been implemented in almost any usual 
programming language. However 
the algorithm requires
$\Omega(n \log n)$ bits on $n$-vertex graphs even if we use a
space-efficient DFS. 
The idea of our algorithm is to classify the edges
via a DFS as tree and back edges and to mark those tree edges
that build a cycle with a back edge. The marking allows us subsequently to
determine the cut vertices and the biconnected components. Given a graph
with $n$ vertices and $m$ edges, the whole
algorithm runs in 
 $O(n+m)$
 time using $O(n+\min\{m,n\log \log n\})$ bits, which is $O(n)$ in
 sparse graphs.

Banerjee et al.~\cite{BanCR16} independently discovered a different
 approach to compute the cut vertices of a graph with $n$ vertices and $m$
 edges that uses $O(n+m)$ bits and time. They have no algorithm with a space
 bound that only depends on $n$.

Finally we study the recognition of outerplanar graphs, i.e., those graphs having a
planar embedding with all vertices on the outer face.  The  problem has been 
studied in various settings and linear time algorithms have been given,
e.g., by
Mitchell~\cite{Mit79} and Wiegers~\cite{Wie86}.  However, both
algorithms modify the given graph by removing vertices of degree $2$, which is not possible in our model. An
easy solution would be to copy the given graph in the working memory, but
this requires $\Omega(n \log n)$ bits for a graph with $n$ vertices. Another
problem is that if the neighbors of a removed vertex are not adjacent, then
both algorithms above want to add a new edge connecting the neighbors. Storing all these
new edges also
can require $\Omega(n \log n)$ bits.
 Our algorithm runs in time $O(n\log\log n)$ and uses 
 $O(n)$ bits, and determines if the input graph is outerplanar, as well
 as if it is maximal outerplanar. 

To obtain our algorithm, 
 we can not simply remove vertices of degree $2$. 
 With each removed vertex $v$
 we have to remove the
 so-called chain of vertices of degree $2$ that contains $v$ and we have to choose the chains
  carefully such that we have only very few %
new edges at a time in our
 graph. %

\section{Preliminary}

For graph-theoretic notions not defined in the paper we
refer to the monograph of Diestel~\cite{Die12}.
For basic notions in the description and the analysis of algorithms---%
e.g., computing {\em tree-} and {\em back edges} with a {\em depth-first
search} (DFS)---we refer to the textbook of Cormen et al.~\cite{CorLRS09}.
To develop space-efficient algorithms, the details of the
representation of an input graph are more important 
than in the classic setting because it is rarely
possible to modify and store a given representation.
We use the terminology of~\cite{ElmHK15}.
In particular, if we say that a graph is 
represented via {\em adjacency arrays}, then we assume that, 
given a vertex $u$ and an index $i$, we can determine the $i$th edge
$\{u,v\}$ of
$u$ in constant time. Moreover, {\em cross pointers} allow us to determine 
the index of $\{u,v\}$ in the adjacency array of $v$ in constant time.
As usual, we always assume that an $n$-vertex graph has vertices $V=\{1,\ldots,n\}$.

Our algorithms make
use of \emph{rank-select data structures}. 
A rank-select data structure is initialized on a bit sequence
$B=(b_1,\ldots,b_n)$ and then
supports the following two queries.
\begin{description}
\item[${rank}_B(j)$] ($j\in\{0,\ldots,n\}$):
Return $\sum_{i=1}^j b_i$
\item[${select}_B(k)$] ($k\in\{1,\ldots,\sum_{i=1}^n b_i\}$):
Return the smallest $j\in\{1,\ldots,n\}$
with ${rank}_B(j)=k$.
\end{description}
Rank-select data structures
for bit sequences of length~$n$ that support
rank and select queries in constant
time and occupy $O(n)$ bits can
be constructed in 
$O(n)$ time~\cite{Cla96}.

Assume that
$d_1,\ldots,d_n\in I\!\!N_0$
and that it is desired to allocate $n$ bit strings 
$A_1,\ldots,A_n$ such that, for $k=1,\ldots,n$,
$A_k$ consists of $d_k$ bits
and (the beginning of) $A_k$
can be located in constant time.
Take $N = \sum_{j=1}^{n} d_{j}$.
We say that $A_1,\ldots,A_n$ are stored
with {\em static space allocation}
if we allocate $A_1,\ldots,A_n$ within
an array of size~$N$.
In $O(n+N)$ time, we can compute the sums
$s_k=k + \sum_{j=1}^{k-1} d_j$ for $k=1,\ldots,n$ and
a rank-select data structure for the bit vector $B$ 
of size $n+N$ whose $i$th bit, for $i=1,\ldots,n+N$,
is $1$ exactly if 
$i=s_k$
for some $k\in\{1,\ldots,n\}$.
This allows us, given a $k\in\{1,\ldots,n\}$
to compute the number of bits
used by the arrays $A_{1}, \ldots, A_{k-1}$
and thus the location of $A_{k}$
in constant time
by evaluating ${select}_B(k)-k$.
One application of static space allocation---as already shown
by~\cite{HagKL15}---is to store data for each vertex $v$ consisting of
$O({\mathrm{deg}}(v))$
bits where ${\mathrm{deg}}(v)$ is the degree of $v$.
Given a vertex, we then can locate its data in constant time and the whole
data can be stored with $O(n+m)$ bits.

To maintain subsets of vertices
or of edge indices
we use a data structure
by Hagerup and Kammer~\cite[Lemma 4.4]{HagK16}
that is called a choice dictionary.
\begin{theorem}
 \label{thm:cd}
 Let~$n\in I\!\!N$.
 A choice dictionary
 is a data structure
 that maintains an initially empty subset~$S$
 of~$\{1, \ldots, n\}$
 under ${insertion}$, ${deletion}$, ${membership}$ queries, and
an operation called ${choice}$ 
 that returns an arbitrary element of~$S$.
 The data structure can be initialized in~$O(1)$ time
 using~$O(n)$ bits of working space and supports each of its
 operations in constant time. The choice dictionary can also be extended by
 an operation called          
 ${iteration}$ that returns all elements in $S$ in a time linear in $|S|$.
\end{theorem}

We also use a simplified version
of the ragged dictionary
that was introduced by 
of Elmasry et al.~\cite[Lemma~2.1]{ElmHK15}
and named in~\cite{HagK16}. 

\begin{theorem}\label{thm:ragged-dictionary}
For every fixed $n \in I\!\!N = \{1, 2, \ldots\}$ as well as integers
 $b = O(\log n)$ and $\kappa = O(n/\log n)$,
there is a dictionary
that can store a subset~$A$ of $\{1, \ldots, n\}$
with~$|A| \le \kappa$,
each $a \in A$ with a string $h_{a}$
of satellite data of $b$~bits,
in $O(n)$ bits
such that the following operations all run in $O(\log\log n)$ time:
$h_{a}$ can be inspected 
for each $a \in A$ and
elements with their satellite data
can be inserted in
and deleted from $A$.
\end{theorem}

\begin{proof}
 The membership test
 can be realized using a bit vector of~$n$ bits.
 It remains the insertion,
 deletion and lookup of satellite data.
 We partition the set
 into $\Theta(n/\log n)$ subsets
 of~$O(\log n)$ elements each
 and maintain an AVL tree for every subset.
 An empty AVL tree
 uses $O(\log n)$ bits.
 Additional $O(\log n)$ bits are used
 for every node
 that is inserted.

 Each AVL tree
 is responsible
 for a subset of $\Theta(\log n)$
 possible keys; more exactly,
  the first AVL tree
 is responsible
 for the keys $\{1, \ldots, \log n\}$
 if they are ever inserted
 into the ragged dictionary,
 the second AVL tree
 is responsible
 for the keys $\{\log n + 1, \ldots, 2\log n\}$
 and so on.
 Hence,
 the time to search for a key
 is 
 $O(\log\log n)$.
 Each node~$u$ consists of a key,
 the satellite data
 that belongs to the key
 and two pointers
 that represent
 the left and the right child
 of~$u$.
 The nodes in the AVL tree
 are stored in an array~$D$
 of $\kappa$ values
 of $\Theta(\log n)$ bits each.
 The positions within~$D$
 that are currently not used
 are maintained within a choice dictionary.
\end{proof}

\section{Cut Vertices}
\label{sec:cv}

A {\em cut vertex\,} of a connected undirected graph $G$ is a vertex 
$v$ such that $G-v$ is disconnected. Furthermore a graph is
{\em biconnected} if it is connected and does not have a cut vertex. 
We first show how 
to compute cut vertices
in~$O(n+m)$ time using~$O(n+m)$ bits
on an undirected graph~$G=(V,E)$
with $n=|V|$ and $m=|E|$. 
Afterwards, we present a second algorithm
that has the same running time
and uses $O(n\log\log n)$ bits.

We start with the description
of an algorithm, but for the time being, do not care
on the running time and
the amount of working space.
Using a single DFS
we are able to classify all edges
as either tree edges or back edges.
During the execution
of a DFS
we call a vertex white
if it has not been reached
by the DFS,
gray if the DFS
has reached the vertex,
but has not yet retracted from it
and black if the DFS
has retracted from the vertex.
W.l.o.g.,
we assume
that all our DFS runs
are deterministic
such that every DFS run
explores the edges of~$G$
in the same order.
Let~$T$ denote
the {\em DFS tree} of~$G$, i.e., the subgraph of $G$ consisting only of tree
edges,
which we always assume
to be rooted at the start vertex of the DFS.
We call a tree edge~$\{u,v\}$ of~$T$
with~$u$ being the parent of~$v$
{\em \mbox{full marked}}
if there is a back edge
from a descendant of~$v$
to a strict ancestor of~$u$,
{\em \mbox{half marked}}
if it is not {\mbox{full marked}}
and there exists a back edge
from a descendant of~$v$
to~$u$,
and {\em \mbox{unmarked}}, otherwise.
Then one can easily prove
the next lemma.
\begin{lemma}\label{lem:find-cut}
 Let~$T$ denote a DFS tree
 of a graph~$G$
 with root~$r$,
 then the following holds:
 \begin{enumerate}
  \item Every vertex~$u \neq r$
   is a cut vertex of~$G$
   exactly if
   at least one of the edges
   from~$u$ to one of its children
   is either an {\mbox{unmarked}} edge or a {\mbox{half marked}} edge.
  \item The vertex~$r$
   is a cut vertex of~$G$
   exactly if it has
   at least two children
   in~$T$.
 \end{enumerate}
\end{lemma}

\begin{proof}
 We consider a vertex~$u$
 with a child~$v$.
 To proof the first part
 we assume that~$u \neq r$.
 The edge~$\{u,v\}$ is {\mbox{full marked}}
 exactly if there is a back edge
 from~$v$ or a descendant of~$v$
 to a strict ancestor of~$u$.
 If the edges from~$u$
 to all of its children
 are {\mbox{full marked}},
 then~$u$,
 the children of~$u$,
 and the parent of~$u$
 all belong
 to the same biconnected component
 and~$u$ can not be a cut vertex.
 If~$\{u, v\}$ is
 either
 {\mbox{unmarked}} or {\mbox{half marked}},
 then there is no back edge
 from~$v$ or a descendant of~$v$
 to a strict ancestor of~$u$
 and the only path from~$v$
 to the parent of~$u$ is through~$u$,
 thus~$u$ is a cut vertex.

 For the second part of the lemma
 we observe
 that the tree edges
 that are incident to~$r$
 can not be {full marked}
 since~$r$ does not have an ancestor.
 If~$r$ has only one child~$v$,
 $r$ can not be a cut vertex
 because, if~$r$ is removed from the graph,
 all remaining vertices 
 still belong to the same connected component
 since they belong to the same subtree of~$T$
 that is rooted at~$v$.
 Otherwise,
 if~$r$ has multiple children in~$T$,
 we denote the first and the second child
 of~$r$
 to be explored by the DFS
 that constructed~$T$ as~$v$ and~$w$, respectively.
 The only path from~$v$ to~$w$
 is through~$r$,
 because otherwise~$w$
 would have been explored
 before the DFS
 retreated from~$v$ back to~$u$.
\end{proof}

Since we want to
use the lemma above,
we have to mark the tree edges
as {\mbox{full marked}}, {\mbox{half marked}} or {\mbox{unmarked}}.
Therefore,
we have two DFS runs.
In the first run,
we classify all tree edges, which we initially unmark.
During the second,
whenever we discover a back edge~$\{w, u\}$
with~$w$ being a descendant of~$u$,
it becomes evident
that both~$u$ and~$w$
belong to the same biconnected component
as do all vertices
that are both,
descendants of~$u$ and ancestors of~$w$,
since they induce a cycle~$C$.
Let~$v$ be the ancestor of~$w$
that is child of~$u$.
We mark the edge from~$u$ to~$v$
as {\mbox{half marked}}
and all the other tree edges on~$C$
as {\mbox{full marked}}.
If the edge~$\{u,v\}$
has already been {\mbox{full marked}}
in a previous step,
this will not be overwritten.
Note that
if~$\{u,v\}$ is {\mbox{half marked}}
and~$u$ has a back edge
to one of its ancestors
such that the edge connecting~$u$
to its parent~$p$
becomes {\mbox{full marked}},
$v$ and~$p$ do not belong
to the same biconnected component,
but~$u$ belongs to both,
the biconnected component containing~$v$
and the biconnected component containing~$p$.
The notion to distinguish
between {\mbox{half marked}} and {\mbox{full marked}}
is to indicate this gap
between biconnected components.

After the second DFS
each cut vertex can be determined
by the edge markings
of the tree edges
connecting it to its children.
For a space-efficient implementation,
the first DFS can be taken from~\cite{ElmHK15}. The second one with 
making the markings is described in the next subsection.
To obtain a space that does not depend on the number of edges,
we combine the two DFS executions
to one DFS execution (Subsection \ref{sec:cv:sea-n}).

\subsection{Cut Vertices with~\texorpdfstring{$O(n+m)$}{O(n+m)} bits}
\label{sec:cv:sea-mn}
We use static space allocation
to address $O(\mathrm{deg}(v))$ bits for each vertex~$v$
of degree~$\mathrm{deg}(v)$.
Within these bits,
we store for every edge~$\{u,v\}$
adjacent to~$v$ if
(1) it is a tree or back edge, (2)
$u$ or~$v$ is closer
to the root
of the DFS tree,
and (3) its markings in case 
it is a tree edge.
Additionally,
using~$O(\log deg(v))$ bits,
when a vertex~$v$ is encountered
by the DFS
for the first time,
we store the position
of the tree edge~$\{u,v\}$
in the adjacency array of~$v$,
where~$u$ denotes the parent of~$v$
in the DFS tree.
This information
can later be used,
when the DFS retreats from~$v$,
such that we can perform
the DFS
without explicitly having to store
the DFS stack.
(This idea to obtain a DFS in $O(n+m)$ time using
$O(n+m)$ bits was already described by Hagerup et al.~\cite{HagKL15}.)
To bound the time
of marking the tree edges
during the second DFS 
by~$O(n+m)$,
we perform the following steps, which are also sketched in Fig.\ref{fig:cv:m1}.
Whenever we encounter a vertex~$u$
for the first time
during the second DFS
via a tree edge,
we scan its entire adjacency array
for back edges
that lead to descendants $w$ of~$u$.
For each such back edge~$\{u,w\}$,
we perform the following substep.
As long as there is a tree edge
from~$w$ to its parent~$v \neq u$
that is not {\mbox{full marked}},
we mark~$\{v,w\}$ as {\mbox{full marked}}
and continue
with~$v$ becoming
the new~$w$.
If we encounter an edge~$\{v,w\}$
that is already marked as {\mbox{full marked}},
we terminate the substep
without marking any more edges.
If at some point the parent $v$ of $w$ becomes $u$,
we mark~$\{v,w\}$ as {\mbox{half marked}}
if it has not been {\mbox{full marked}}, yet,
and terminate the substep thereafter.

It is easy to see
that this procedure
results in the correct markings
for every edge.
The order in which back edges are worked on
assures whenever a tree edge~$\{v,w\}$
is already {\mbox{full marked}}
so are all tree edges
between~$v$ and the child of~$u$.
Because we store,
for each vertex,
the position of the edge
that connects it
with its parent in~$T$,
the time of each substep
is at most the number of tree edges
that are marked
plus additional constant time.
Since each tree edge
is marked at most twice,
once as {\mbox{half marked}}
and once as {\mbox{full marked}},
and the number of back edges
is at most~$m$,
altogether 
we use $O(m)$ time.

\begin{figure}[h]
 \centering%
 \includegraphics[scale=0.70]{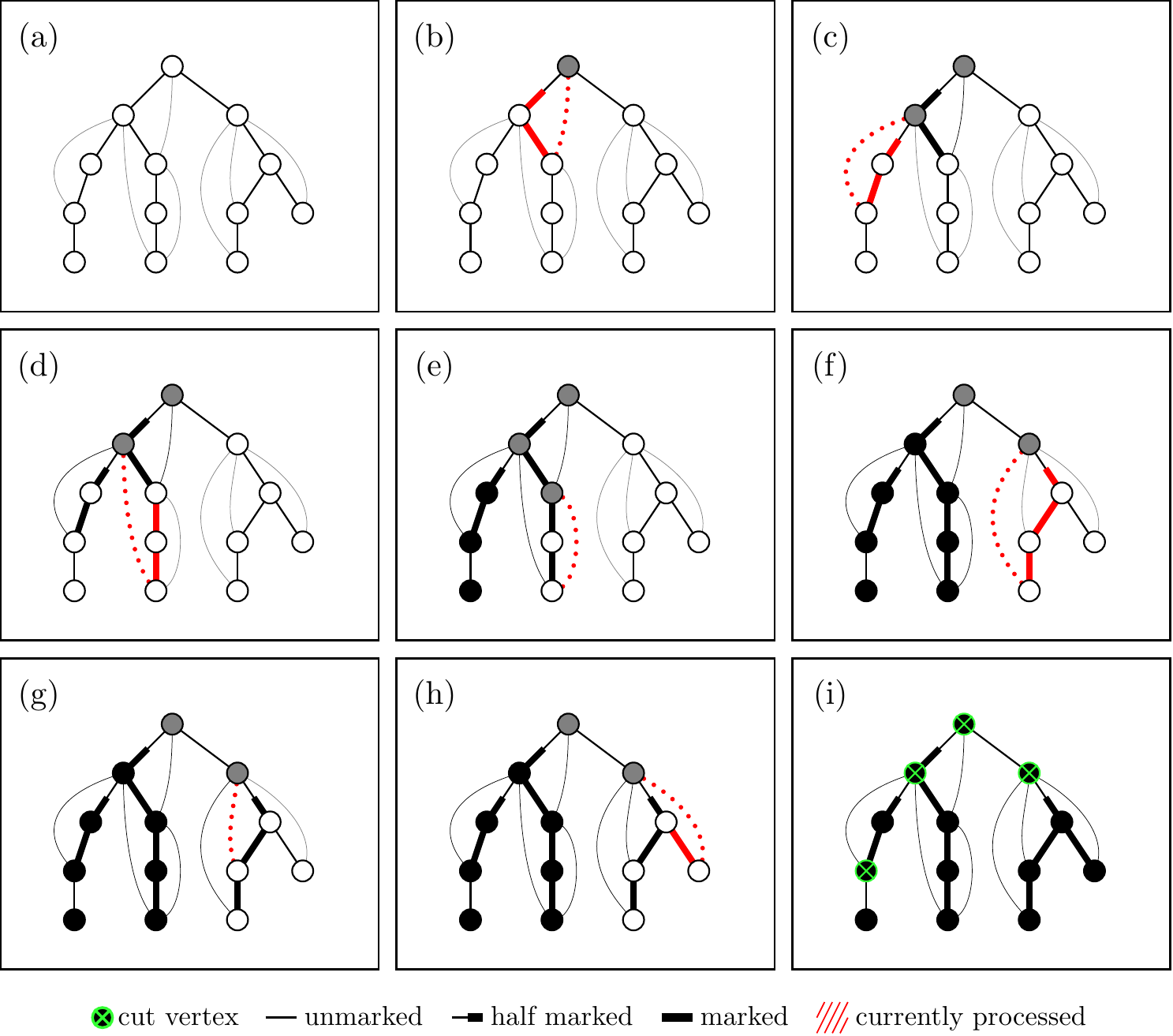}%
 \caption{%
Assume that a first DFS has marked the edges as tree- or back edges. 
   Snapshots are shown while we
 run a second DFS---the process of the DFS is shown by the vertex
 colors---and while we visit a vertex $v$, we iterate over all back edges
 connecting $v$ and a descendent $u$ of $v$. Whenever we consider a back edge
 (shown dotted), we mark all tree edges on the path from $v$ to $u$ (edges
 get bold). Note that in (d), we stop the marking when reaching the first
 marked edge. In (e) and (g), we have no new
 markings at all. The cut vertices are determined in (i).
 }
 \label{fig:cv:m1}%
\end{figure}

\subsection{Cut Vertices with~\texorpdfstring{$O(n\log\log n)$}{O(n*loglog(n))} bits}
\label{sec:cv:sea-n}
If the number of edges is sufficiently high,
the limitations on the working space
neither allow us
to mark all edges as tree edges or back edges at once
nor to store for every vertex~$v$
the position of the tree edge~$\{u,v\}$
to its parent~$u$
in~$T$. 
The key ideas
to make the algorithm more space-efficient
are to perform
the classification of edges
as tree edges or back edges
on the fly
whenever an edge
is explored
by the DFS,
to use 
a second stack~${\cal U}$,
and to apply the stack restoration technique
by Elmasry et al.~\cite{ElmHK15}
to both stacks. We thus 
reconsider that paper.
To obtain a DFS with a linear running time using $O(n \log \log n)$ bits,
the stack is partitioned into $O(\log n)$ segments
of $\Theta(n/\log n)$ entries, each consisting basically of a vertex. 
During the DFS, only the $O(1)$ latest
segments are kept in the working memory;
the remaining are thrown away.
Whenever a segment that was thrown away is needed,
a restoration recovers it in a time linear in the number of vertices of the segment. 
To restore a segment in a time linear to its size,
the vertices of the $i$th segment
have a {\em hue} (value) $i$.
For a slight modification of that algorithm,
one can easily see that it is not important
that the segments consist of $\Omega(n/\log n)$ vertices stored in the
stack, as long as their number
is bounded by $O(\log n)$ and 
the vertices with the same hue 
form a connected subsequence
of the sequence of elements
in the stack. 
Therefore, we build segments not based on the
entries of the stack; instead, we define the first $\Theta(n/\log n)$
vertices that are 
visited by the DFS as the first segment, the next $\Theta(n/\log n)$
vertices as the next segment, etc.
The hue values are determined via an extra DFS at the
beginning. They are stored 
in addition to the colors white, gray and
black that are used
during a ``standard'' DFS.
The space consumption of the algorithm sums up %
to $O(n)$ for the
segments plus
$O(n\log\log n)$ bits for the hue. %

A normal DFS stack~${\cal S}$
contains at any moment
during the execution
of a DFS
the vertices
on the path within~$T$
from its root
to the vertex
that is currently explored,
which are all the vertices
that are currently gray.
The depth $dp_{w}$ of a vertex~$w$
is the number of edges connecting the vertices
on the path from~$w$ to the root~$r$
of~$T$
that consists solely of tree edges.
The second stack~${\cal U}$
contains those gray vertices~$u$
that have a gray child~$v$
such that~$e_{u}=\{u,v\}$
is not {\mbox{full marked}}.
Hence, the vertices in~${\cal U}$
denote a subset of the vertices in~${\cal S}$.
More exactly, each entry of~${\cal U}$
is a tuple of a vertex~$u$
and its depth within~$T$.
When~$v$ is explored,
$v$ is pushed onto~${\cal S}$
and its parent~$u$ (together with its depth)
is pushed onto~${\cal U}$.
When the DFS
retreats from~$v$,
it is popped from~${\cal S}$
and~$u$
is popped from~${\cal U}$
if it is still present in~${\cal U}$.
A second way for~$u$
to be removed from~${\cal U}$
is when the edge~$e_{u}$
becomes {\mbox{full marked}}. %

Let~$w$ denote the vertex
at the top of~${\cal S}$.
Whenever the DFS
tries to explore a vertex~$u$
that is gray,
then~$\{w,u\}$ is a back edge.
Under the assumption
that we know
the depth of every vertex, 
we can mark edges as follows:
While there is a vertex~$u'$
on the top of~${\cal U}$
whose depth is higher than~$dp_{u}$,
we full mark the edge~$e_{u'}=\{u',v'\}$
that connects~$u'$ with its gray child~$v'$
and pop~$u'$ from~${\cal U}$.
This loop stops
if either~$u$ becomes the vertex
at the top of~${\cal U}$
or a vertex with a lower depth than~$u$
becomes the top vertex of~${\cal U}$.
In the first case,
we half mark~$e_{u}$.
In the second case,
we do not mark the edge~$\{u,v\}$
that connects~$u$ with its gray child~$v$
since $u$ is not on~${\cal U}$
because the edge~$\{u,v\}$
must have been {\mbox{full marked}}
during the processing
of a previous back edge.

\begin{figure}[b]
 \centering%
 \includegraphics[scale=0.70]{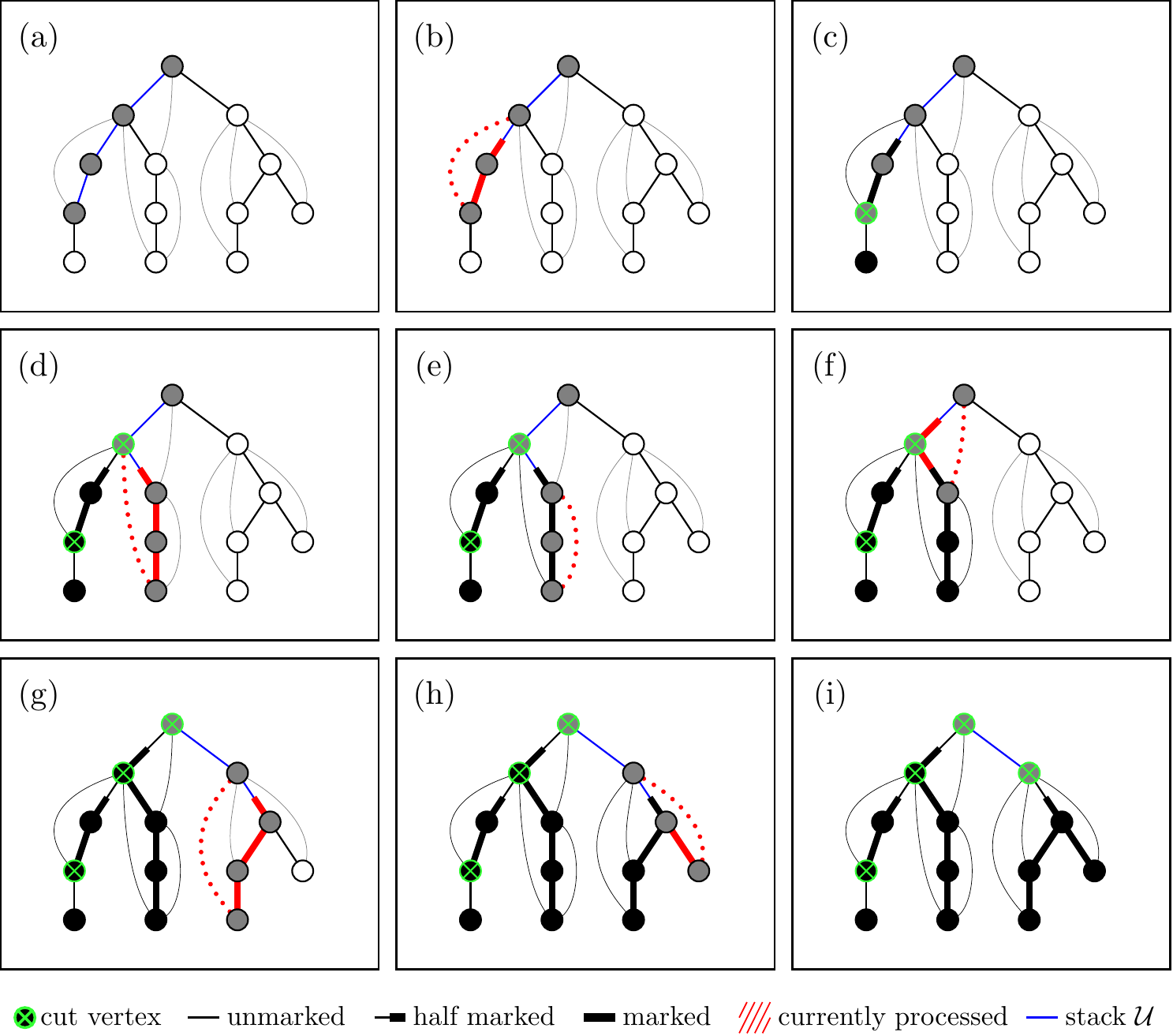}%
 \caption{%
  Snapshots of the DFS
  that determines which vertices are cut vertices.
  (a) shows the graph right before the first back edge is processed.
  All the upper endpoints of tree edges
  that are on the path from the root
  to the vertex that is currently processed
  are part of the second stack ${\cal U}$.
  In (b) the back edge is processed,
  which results in the removal
  of the (upper) endpoints of {\mbox{full marked}} tree edges from ${\cal U}$
  and the tree edge that belongs to the thereafter topmost vertex in ${\cal U}$
  becoming {\mbox{half marked}}.
  In (c) the DFS retreats over an {\mbox{unmarked}} tree edge
  to a non-root vertex $u$,
  and $u$ is identified as a cut vertex.
  Later in (d) the DFS has retreated over a {\mbox{half marked}} tree edge,
  which resulted in another vertex being identified as a cut vertex.
  The root itself is determined to be a cut vertex
  as soon as a second subtree is explored (in (g)).}
  \label{fig:cv:m2}%
\end{figure}

We now discuss
the computation of cut vertices, which is also sketched in Fig.~\ref{fig:cv:m2}.
When retreating from a vertex~$v$
to its parent~$u$
that is not the root of~$T$,
we check
if~$\{u,v\}$ is {\mbox{half marked}} or {\mbox{unmarked}}.
If that is the case,
we output~$u$ as a cut vertex.
For the root vertex~$r$,
we maintain a counter
that indicates if
at least two children of~$r$
have been explored
during the DFS.
If so,
$r$ is outputted as a cut vertex.
Using a bit vector
over the vertices,
we can easily avoid
outputting a cut vertex more than once.

For the implementation of the algorithm,
the remaining problems are 
maintaining both stacks~${\cal S}$ and~${\cal U}$ as well as
determining the depth of~$u$,
whenever processing a back edge~$\{w, u\}$
with~$u$ being the ancestor of~$w$
in the DFS tree~$T$.
For the first problem,
we store for every vertex
its hue and
use the stack restoration techniques
introduced by Elmasry et al.~\cite{ElmHK15},
but use it with the modified 
size
as described above.
Restorations of segments in~${\cal S}$ and~${\cal U}$
are performed independently of each other.

The second problem is more complicated and considered now.
Let $k$ be the hue of vertices that we currently process, and let 
$Z$ be the set consisting of
every hue~$i\neq k$ that is 
present in~${\cal U}$. 
Our goal is to store the depth of all vertices
in ${\cal U}$ with 
a hue in $\mathrm{max}(Z) \cup \{k\}$
such that it can be addressed by the vertex.
The idea is to use one array $A$ 
addressed with static space allocation to store the depth
for all vertices in ${\cal U}$ of one segment, i.e., one hue.
 This allows us to build and destroy
the arrays for each hue independently. However, the rank-select
 structure used by the static space allocation 
is to slow in its construction since we want a running time of
 $O(n/\log n)$. Therefore, we 
define {\em blocks} where block 
$i\in \{0,\ldots,\lceil n/\lceil (\log n)/2 \rceil\rceil-1\}$
consist of the vertices
$1+i\lceil (\log n)/2 \rceil ,\ldots,(i+1) \lceil (\log n)/2 \rceil$---the last block may be smaller.
In an auxiliary array $B$
we store, for each block $b$,
a pointer to the first entry in $A$
that contains the depth %
regarding a vertex in $b$.
For each vertex within a block $b$, we use table lookup to find the
position relative to the first
entry for a vertex in $b$.
This allows us to store 
the depth of the vertices in the upcoming segment
in an array with static space allocation.
We also 
restore the depth
of the %
vertices of a segment in ${\cal U}$
whenever we restore it. %

When processing a back edge~$\{u, w\}$
($w$ having a hue~$k$)
with~$u$ being the ancestor of~$w$
and having a hue~$i$
there are two possibilities:
If~$i\in \mathrm{max}(Z) \cup \{k\}$, 
then the depth of~${\cal U}$
can be determined in constant time.
Otherwise, $i < j= \mathrm{max}(Z)$
and,
we iteratively restore those segments in $\{j-1, j-2, \ldots, i\}$
that have a vertex in~${\cal U}$
and process the vertices within these segments
that are present in~${\cal U}$
as described above
until we finally restore
the segment of vertices with hue~$i$
together with their depth.

We summarize the space bound as follows.
For every vertex
we store in~$O(n)$ bits
its current color (white, gray, black), and
if it yet has been outputted
as a cut vertex. The set $Z$ can be easily maintained with $O((\log n)^2)$ bits.
Using~$O(n\log\log n)$ bits
we can store for every vertex
its hue.
We use additional~$O(n)$ bits
to store a constant number of stack segments
of~$O(n/\log n)$ vertices each.
Since the depths of those vertices are stored
compactly in arrays, each such array $A$ together with its auxiliary
array $B$ can be implemented
with $O(n)$ bits.
Moreover, 
we use two bits for each vertex $u$ on the
stack to store if the edge connecting $u$ and its child on the stack, if any, has
been half or fully marked. 
Thus, the overall space bound is $O(n\log\log n)$ bits.

We finally determine the time bounds.
As analyzed in~\cite{ElmHK15}, the DFS including the stack restorations of~${\cal S}$,
but without the extra computations for the back
edges
runs in $O(n+m)$ time
as do the intermediate computations in total.
Assume for the moment, that we do not throw away and restore segments of
${\cal U}$.
Then, the total time to
mark the tree edges due to the back edges is $O(n+m')$ where $m'$ is the
number of back edges since each tree edge is marked at most twice using the
stack ${\cal U}$. The tests if a vertex is a cut vertex can be performed in
total time $O(m)$.
It remains to bound the time for the restorations of ${\cal U}$. 
Whenever we restore a segment of ${\cal U}$, a hue value is removed from the
set $Z$ and never returns. This means that we have only $O(\log n)$
restorations of %
a segment of ${\cal U}$, which can be done in total time
$O(n)$. %

Combining the algorithms of Section~\ref{sec:cv:sea-mn} and
this section, we obtain our first
theorem.
\begin{theorem}
 There is an algorithm that,  
 given an $n$-vertex $m$-edge graph~$G$
 in an adjacency array representation
 with cross pointers, runs in time $O(n+m)$ and uses $O(n+\min\{m,n\log \log n\})$ bits
 of working space, and
 determines the cut vertices of $G$.
\end{theorem}

\subsection{Biconnected Components}
\label{sec:bc}

We next show that we can compute the biconnected components of an undirected
graph. 
Recall that a graph is biconnected if it is connected 
and has no cut vertices.
A biconnected component of a graph $G$ is a maximal biconnected  
induced subgraph of $G$.  
Whenever we say
in the following theorem and proof that we output an edge $\{x,y\}$, then we mean that we 
output the
index of the edge in the adjacency array of both $x$ and $y$.

\begin{theorem}
There is a data structure that, given an $n$-vertex $m$-edge graph $G$ 
in an adjacency array representation
 with cross pointers,
runs $O(n+m)$ initialization time and uses $O(n+m)$ bits of working space
and that afterwards, given an edge $e$, computes the vertices and/or the edges of the biconnected
component~$B$ of $G$ with $B$ containing $e$ in 
a time that is linear in the 
number of vertices and edges that are output.
\end{theorem}

\begin{proof}
To initialize the data structure, first run our algorithm from
Section~\ref{sec:cv:sea-mn} to compute a DFS tree $T$
with a root $r$, for each edge $\{u,v\}$
the ancestor-descendant relationship in $T$ between $u$ and $v$ as well as the markings
of the tree edges. 
Also store for each vertex $v\neq r$ the index of its edge connecting $v$ to
its parent.
Then build a rank-select
data structure for each vertex $v$ that allows us to iterate over the fully
marked tree edges and the back~edges to ancestors of $v$ in $O(1)$ time per edge.
It is easy to see that the initialization runs in $O(n+m)$ time and all
information can be stored with $O(n+m)$ bits
using static-space allocation.

Afterwards, given an edge $e=\{u,v\}$
with $u$ being an ancestor of $v$, we can output the biconnected
component $B$ containing $e$ by 
running a DFS from $v$ and traversing only
tree edges (to both directions, parent and children) 
such that 
we never explore new vertices from a vertex that was reached by a half
marked or unmarked edge and such that we never
move via an unmarked or half marked edge from a parent to its child.
During the DFS, we output all visited vertices as well as all traversed tree
edges.
In addition, whenever we visit a vertex $v$ that was reached via a fully
marked edge, we output all back edges $\{u,v\}$ with $u$ ancestor of
$v$.
The algorithm can easily be modified such that it only outputs the
vertices/edges of
the biconnected component.
Using the rank-select data
structures, this can be done in the time stated in the lemma.

To see that each outputted $B$ is indeed a biconnected component observe that 
$B$ is connected. Assume that 
$B$ has a 
cut vertex $v$ with a child $u$ that cuts off 
the subtree $T_u$ with root $u$, and we want to output the component containing
the edge connecting $v$ and its parent. 
In that case the edge $\{v,u\}$ is half marked or not marked. 
Hence, such an edge 
is not used to go from a parent to a child. On the other hand, if 
we want to output a component $B$ 
for which a vertex $v$ is a cut vertex disconnecting $B$ and $r$, then the edge 
connecting $v$ with the rest of $B$ is half marked.
Hence, we output $v$, but we do not explore
any other edges from $v$.

If a vertex $v$ was reached by a fully marked edge, then
the edge from $v$ to its parent belongs the currently outputted biconnected
component $B$. Since each back edge $\{u,v\}$ with $u$ ancestor of $v$ always
belongs to the biconnected component $B'$
that contains all tree edges $u$ to $v$,
$B=B'$ exactly if $v$ was reached by a fully marked edge.

Finally, $B$ is maximal by construction
since all tree edges
that are in $B$ are fully marked
except for the ``highest'' edge in $T$
that is half marked; 
thus all vertices of $B$ are found by the DFS from above 
starting at $v$.
\end{proof}

\section{Outerplanar Graphs}

Outerplanar graphs and
maximal outerplanar graphs  are well-studied  
subclasses of planar graphs.  For  their structural 
properties we refer to the monograph of 
Brandst\"adt et al.~\cite{BraLS99}. 

Given a biconnected outerplanar graph~$G=(V,E)$,
we call an edge
that is incident to the outer face
an \emph{outer edge}
and an edge
that is incident to two inner faces
an \emph{inner edge}.
We now describe
a set of well-known properties
for outerplanar graphs
that help us
to describe and prove our algorithm
of Section~\ref{sec:bop}. 
Every maximal outerplanar graph
with at least three vertices
is biconnected
and, for every biconnected outerplanar graph~$G$,
the set of outer edges
induces a unique Hamiltonian cycle
that contains all vertices of~$G$.
Every biconnected outerplanar graph~$G=(V,E)$
with $V = \{1, \ldots, n\}$ 
and $|E| > n$
contains at least one inner edge.
Let the vertices of~$G$ be labeled
according to their position
on the Hamiltonian cycle of~$G$,
by $1, \ldots, n$,
and let~$\{u,v\}$ denote an inner edge
that connects the vertices~$u$ and~$v$.
W.l.o.g., let $u < v$.
Then the graph $G'=G[\{u, \ldots, v\}]$
is biconnected and outerplanar.
Since $\{u,v\}$ is an inner edge,
$1 < v-u < n-1$ holds
and there are exactly $v-u-1>0$ vertices
between~$u$ and~$v$
on the path part of the Hamiltonian cycle
that belongs to~$G'$.
This path together with the edge $\{v,u\}$
forms the Hamiltonian cycle of~$G'$.

Usually,
one decomposes an outerplanar graph
by repeatedly removing
a vertex~$v$ of degree~$2$.
However,
this needs to test
if the neighbors of~$v$
are connected by an edge
and, if not, to add
such an edge.
Because the test 
is too time consuming
and 
storing all such edges needs
too much space,
we search instead
for a closed or good chain
defined next.

We define a \emph{chain}\footnote{\cite{Sys79}
uses instead of chain
the term \emph{maximal series of edges}.}
in an outerplanar graph~$G$
as either a cycle
that consists solely of vertices of degree~$2$
or a path
that contains
at least three pairwise distinct vertices
with the property
that its first and its last vertex
have a degree larger than~$2$
while the rest 
must have degree~$2$.
We denote the first and the last vertex of a chain
as its \emph{endpoints}
unless the chain~$C$ is a cycle,
in which case
the endpoints can be chosen arbitrarily
as long as they are adjacent to each other.
Furthermore,
we call a cycle a \emph{loop}
if it contains one vertex of degree larger than~$2$
and all 
other vertices have a degree~$2$.
A chain
is called a \emph{good chain}
if one of its endpoints has a degree
of at most~$4$.
Let us call a face~$F$
\emph{induced} by a chain~$C$
if the endpoints~$u$ and~$v$ of~$C$
are adjacent to each other
and~$C$ together with the edge~$\{u,v\}$
is the boundary of~$F$.
We denote a chain~$C$
that induces a face~$F$
as a \emph{closed chain}.
For simplicity,
we sometimes consider a chain
also as a set of edges.

\begin{lemma}\label{lem:bop:good}
 Let~$G=(\{1, \ldots, n\},E)$ denote
 a biconnected outerplanar graph
 with $n \ge 3$ vertices.
 Then~$G$ contains a good closed chain~$C$.
\end{lemma}

\begin{proof}
 If~$G$ is a cycle
 the whole graph is a good closed chain.
 Otherwise,
 $G$ has at least one inner edge. 
 W.l.o.g., assume that vertices are numbered
 such that the vertices of the Hamiltonian cycle of~$G$
 are $1, \ldots, n$ in this order.
 We first assume
 that there are no two vertices of degree~$2$
 adjacent to each other
 and prove
 that then there is always a good chain~$C$
 with a vertex~$v_{g}$ of degree~$2$
 adjacent to an endpoint~$v_{e}$ of degree~$\le 4$
 such that~$v_{g},v_{e} \in \{2, \ldots, n-1\}$.
 We prove the lemma by induction.
 The base case
 is a biconnected outerplanar graph
 with four vertices and five edges.
 Depending on the start of the numbering
 either vertex~$2$ or~$3$
 is a vertex of degree~$2$
 and the other one has degree~$3$
 and is the endpoint of a good closed chain
 containing the first one.
 We now consider
 a biconnected outerplanar graph~$G$
 with~$n>4$ vertices
 and assume
 that the lemma holds
 for any graph
 with at least~$4$
 and less than~$n$ vertices.
 We consider an arbitrary vertex~$u \in \{2, \ldots, n-1\}$
 of degree~${\ge 3}$.
 There are two cases to consider:
 If~$u$ has a \emph{long edge}~$\{u,v\}$
 that is an inner edge to a vertex~$v$
 such that~$|u-v| > 2$,
 then the set~$\{u, \ldots, v\}$
 induces a biconnected outerplanar subgraph~$G'$ of~$G$
 of at least four vertices.
 Since~$\{u, \ldots, v\}$
 is a proper subset of $\{1, \ldots, n\}$,
 $G'$ has less vertices than~$G$.
 The only vertices
 whose degree in~$G'$
 is different from their degree in~$G$
 are~$u$ and~$v$,
 which are now neighbors on the Hamiltonian cycle in $G'$.
 So after 
 a renumbering such that $u$ becomes vertex~$1$
 and $v$ becomes vertex~$n'$,
 we conclude by induction that
 $G'$ and thus~$G$
 contain vertices~$v_{e}$ and~$v_{g}$.
 Notice that,
 if the degree of~$u$ is greater than~$4$,
 then~$u$ has at least one inner edge~$\{u,v\}$
 to an vertex~$v$
 such that~$|u-v| > 2$.
 Hence, in the second case,
 if there is no long inner edge
 incident to~$u$,
 the degree of~$u$ is at most $4$
 and~$u$ has an inner edge
 such that $|u-v|=2$.
 There exists exactly one vertex~$w$
 between~$u$ and~$v$
 on the Hamiltonian cycle.
 Thus, $\{\{u, w\}, \{w, v\}\}$ is a good closed chain.
 We found~$v_{e}=u$ and~$v_{g}=w$.

 When we remove that assumption
 that no two vertices of degree~$2$
 are adjacent to each other,
 the same proof holds
 for general biconnected outerplanar graphs
 with at least one inner edge
 if one uses only one single number
 for the labels of all vertices of degree~$2$
 that belong to the same chain
 on the Hamiltonian cycle.
\end{proof}

The next two lemmas are used to show the correctness of our algorithm.
For each set of edges~$E'$
we define $graph[E']$
as the graph with vertex set $\{v \mid v \textrm{ is endpoint of edge in } E'\}$
and edge set $E'$.

\begin{lemma}
 \label{lem:bop:prop}
 Let~$G=(V,E)$ be
 a biconnected outerplanar graph
 with at least three vertices,
 then the following properties hold:
 \begin{description}
  \item{(i)} For every chain~$C$ in $G$
   with endpoints~$v_{1}$ and~$v_{2}$,
   the graph~$G' := graph[(E\setminus C)\cup  \{\{v_{1},v_{2}\}\}]$
   induced by the edge set $(E\setminus C) \cup \{\{v_{1},v_{2}\}\}$
   is biconnected and outerplanar.
  \item{(ii)} For all $v_{1}, v_{2} \in V$,
   there are at most two internal vertex-disjoint paths
   with at least two edges each.
 \end{description}
\end{lemma}
\begin{proof}
 To prove the first part
 note that~$G'$ is a minor of~$G$.
 Thus, $G'$ is outerplanar.
 For all vertices $u_{1}, u_{2}$ of~$G'$,
 there are two internal vertex-disjoint paths
 in~$G$.
 Since~$C$ is a chain,
 it can only be a complete part
 of such a path and
 therefore replaced by the edge~$\{v_{1}, v_{2}\}$.
 It follows there are
 two internal vertex-disjoint paths
 between~$u_{1}$ and~$u_{2}$ in~$G'$
 and~$G'$ is biconnected.
 We prove the second part by contradiction.
 Assume that Prop.~ii
 does not hold.
 Then $K_{2,3}$ is a minor of~$G$.
 This is a contradiction
 to~$G$ being outerplanar.
\end{proof}

\begin{lemma}
 \label{lem:bop:proof}
 Let~$G=(V,E)$ be a graph
 for which the following properties hold
 for all vertices $v_{1}$ and $v_{2}$
 and all chains~$C$ in~$G$ with endpoints $v_{1}$ and $v_{2}$.
 \begin{description}
  \item{(1)} %
   The graph~$G' := graph[(E \setminus C) \cup \{\{v_{1},v_{2}\}\}]$
   is biconnected and outerplanar.
 \item{(2)} There are at most
   two internal vertex-disjoint paths
   with at least $2$ edges each in~$G$
   that connect~$v_{1}$ and~$v_{2}$.
 \end{description}
 Then~$G$ is biconnected and outerplanar.
\end{lemma}
\begin{proof}
 The existence of a
 chain
 follows from Lemma~\ref{lem:bop:good}.
 $G'$ is biconnected outerplanar
 because of Prop.~1.
 Because of Prop.~2 
 there is at most one internal vertex-disjoint path
 connecting~$v_{1}$ and~$v_{2}$ in~$graph[E\setminus C]$.
 Hence, the edge~$\{v_{1}, v_{2}\}$
 is part of the outer face of~$G'$.
 It follows
 that we can embed~$C$
 in the outer face
 of an embedding of~$G'$
 to yield an outerplanar embedding of~$G$.
 Because~$\{v_{1}, v_{2}\}$
 is part of the outer face of~$G'$,
 it is part of the Hamiltonian cycle of~$G'$.
 We can extend the Hamiltonian cycle of~$G'$
 by replacing~$\{v_{1}, v_{2}\}$
 with~$C$
 to get a Hamiltonian cycle of~$G$.
 Thus~$G$ is biconnected.
\end{proof}

\subsection{Our Algorithm on Biconnected Outerplanar Graphs}
\label{sec:bop}
For the time being, we assume that the given graph is biconnected.
Our algorithm works in two phases
and can be sketched as follows.
Before our actual algorithm starts
we test whether~$m \le 2n-3$.
If not, the graph is not outerplanar
and we terminate immediately.
Otherwise we start our algorithm
that modifies the input graph into a smaller outerplanar graph
as described in Lemma~\ref{lem:bop:proof} Prop.~1
while checking Prop.~2
of Lemma~\ref{lem:bop:proof}.
Lemma~\ref{lem:bop:prop} guarantees
that 
it does not matter which chain we take,
the modification
is always possible
and that the check
never fails unless
$G$ is not biconnected outerplanar.
A main obstacle
is to handle the edges 
replacing the chains
that we call subsequently \emph{artificial edges}.

To check Prop.~2
of Lemma~\ref{lem:bop:proof},
we keep in both phases
counters~$P_{e}$ for every (original or artificial) edge~$e$
to count the number of internal vertex-disjoint paths
with at least $2$ edges
that connect the endpoints of~$e$.
Whenever we remove a chain with both endpoints in $e$, we increment
$P_{e}$.
If~$P_{e}$
at any time exceeds~$2$,
the graph can not be outerplanar
by Lemma~\ref{lem:bop:prop}~(ii).
We check this
after every incrementation
of an counter~$P_{e}$
and terminate eventually. %

We start to sketch the two phases of our algorithm. The
details of the phases are given subsequently.
Let~$G=(V,E)$ denote the input graph,
which is located in read-only input space,
and $G'$ denote the subgraph of~$G$
that we are currently considering
in our algorithm.
Initially, $G' := G$.
Thereafter,
within our algorithm
the edges of chains
are either removed from~$G'$
if they induce a face
or replaced by artificial edges
that connect the endpoints of the chain directly.
Consequently, $G'$ is always a minor of~$G$.

The purpose of the first phase
is to limit the number
of artificial edges
that are required
for the second phase
to $O(n/\log n)$.
The first phase
consists of $\Theta(\log\log n)$ \emph{rounds},
in each of which
we iterate over all chains, but only
remove the closed ones.

In the second phase,
we repeatedly take a vertex of degree 2 and determine its chain. Depending on
the kind of the chain, we proceed:
Good closed chains are removed
from~$G'$.
The edges of chains
that turn out to be
good, but not closed
are replaced
by an artificial edge
connecting its endpoints.
The counter of vertex-disjoint paths
with the same endpoints~$\{u,v\}$
for a newly created artificial edge~$\{u,v\}$
is initialized with~$1$
to account for the chain
that has been replaced by~$\{u,v\}$.
When processing a chain~$C$
that turns out to be not good,
we implement a \emph{shortcut}
that is a pair of pointers,
each one addressed
by the vertex of degree~$2$
in~$C$
that is next to a endpoint of~$C$
and pointing to
the vertex of degree~$2$ in~$C$
that is adjacent to the other endpoint.
This way,
when the degree
of either one of the endpoints
is lowered to~$2$
by the removal of adjacent chains
and thus~$C$
becomes connected with another chain,
$C$ does not have to be traversed again
to check 
if the new chain is good.

If the input graph is outerplanar,
the algorithm terminates
either in Phase~$1$ or in Phase~$2$
as soon as the last good chain,
which is a cycle,
is processed and a single edge
that is the edge between its endpoints
remains.
Otherwise,
if the input graph is not outerplanar,
at some moment 
one of the checks fails
and the algorithm stops 
and answers that the input is not 
biconnected outerplanar.
Possible conditions
for a failed check are
if no good chain remains
and the graph has at least vertices,
if the counter
that counts the internal vertex-disjoint paths
with at least $2$ edges
between the endpoints of an edge
exceeds~$2$ for some edge,
if some loop is detected,
or if a vertex turns out to be incident
to at least three vertices of degree~$2$. %

To lower the time
that is required to iterate through the adjacency array
of a vertex~$v$,
we initialize,
for each vertex~$v \in V$
with degree $deg_{G}(v)$ in~$G$,
a choice dictionary
with universe
$\{1, \ldots, deg_{G}(v)\}$
that represents the edges
that are adjacent to~$v$
and still present within the current subgraph.
The choice dictionary 
contains~$i$ with $i \in \{1, \ldots, deg_{G}(v)\}$
exactly if the $i$th entry of the adjacency array of $v$
represents an edge of the input graph~$G$
that is still present in the current graph~$G'$.
Thus, future iterations
through the adjacency array of~$v$
can be performed within a time
that is linear in the number of edges
that are incident on~$v$ in~$G'$.
Hence, the time to determine
the adjacency of two endpoints~$v$ and~$w$
of a chain
is bounded by $O(\min\{deg_{G'}(v),deg_{G'}(w)\})$.

{\bf Phase~$1$:}
At the beginning of each round
we unmark all vertices,
which might have been marked
during the previous round,
and insert all vertices of degree~$2$
of the current graph~$G'$
into a choice dictionary~$\hbox{\textit{D}}$.

For a more detailed description
we subdivide the $\Theta(\log\log n)$ rounds
into stages.
As long as
there is a vertex~$u$ in~$\hbox{\textit{D}}$,
we extract it from~$\hbox{\textit{D}}$
and start a new stage.
Let~$C$ denote the chain
that~$u$ belongs to.
A stage works as follows.
If~$C$
contains a vertex
that is marked as \hbox{\textit{tried}}
or~$C$ is not closed,
the stage stops.
Otherwise,
we increment the counter~$P_{e}$
for every edge~$e$
on~$C$ as well as
for the edge connecting its endpoints
and remove
the 
vertices and edges
of~$C$
from the graph.
If the current chain~$C$
is not the first chain
processed in the current stage so far,
we mark every vertex of degree~$2$
that has been processed
within the current stage
and still remains in the graph
as~\hbox{\textit{tried}}.
If one or both of the endpoints of~$C$
thereby become 
a vertex of degree~$2$,
the procedure is repeated immediately
for the new chain~$C'$
that incorporates
both endpoints of~$C$.
The stage ends if
no new chain is found.
At the end of the stage
we remove the vertices of degree~$2$
that have been processed
within the current stage
from~$\hbox{\textit{D}}$.
A round ends if $\hbox{\textit{D}}$ is empty.

{\bf Phase~$2$} (illustrated in Fig.~\ref{fig:op:p2}){\bf:}
We 
start to initialize a choice dictionary $\hbox{\textit{D}}$
consisting of one vertex of degree~$2$ for each 
good closed chain.
(During Phase 2, vertices of other chains may be added into $D$.)
Take $q$ as the initial number of vertices in $D$.
As we 
prove in Corollary~\ref{cor:dual:leafbound},
the number of chains
that induce a face
is at most 
$O(n/\log n)$ after Phase~$1$ unless $G$ is not outerplanar.
If~$G'$ contains more such chains,
we can terminate immediately.
Otherwise,
we can store a constant number
of artificial edges and shortcuts
for each such chain.
Whenever we process a vertex~$u$
from $\hbox{\textit{D}}$,
we determine
the chain~$C$
that~$u$ is part of, the endpoints $v$ and $w$ of $C$,
and then distinguish three cases. %
\begin{itemize}
\item If $C$ is good and closed,
we increment the counter~$P_{e}$
for every edge~$e$
that is part of~$C$
and for the edge
that connects its endpoints.
Finally, we remove the edges of~$C$
from~$G'$. 

\item If
$C$ is good and not closed,
we insert an artificial edge~$e_{a}$
that connects the endpoints of~$C$,
increment the counter~$P_{e}$
for every edge~$e$ in~$C$ by~$1$
and 
remove all the edges and inner vertices of~$C$
from~$G'$.
Finally, the counter~$P_{e_{a}}$
of the artificial edge~$e_{a}$
is initialized with~$1$.

\item Otherwise,
$C$ is not good.
We remove all vertices of $C$ from $D$ and
insert a shortcut
between the vertices
of degree~$2$ in~$C$
that are adjacent to
the endpoints of~$C$. 
If there are old shortcuts 
connecting other inner vertices of $C$, they become obsolete and are
removed. %

\end{itemize}
Whenever 
the degree of 
$v$ or $w$
decreases,
we perform the following subroutine. 
If $deg_{G'}(v)\in \{3,4\}$, 
we insert all vertices $u$ of degree~$2$
that are adjacent to~$v$
into~$D$.
Note that $u$ and $w$
are neighbors of~$v$ on the unique Hamiltonian Cycle of~$G'$.
Thus, if~$G$ is biconnected outerplanar, there is at most one such vertex $u$.
If $deg_{G'}(v)=2$,
we insert~$v$ into~$D$.
We proceed with $w$ analogously. We so guarantee that each good closed chain
of $G'$ has a vertex in $D$.
As shown by Lemma~\ref{lem:artifEdges},
if the number of artificial edges and shortcuts 
that are simultaneously in use exceeds $2q$,
$G$ is not outerplanar and we stop. %
\begin{figure}[t]
 \centering%
 \includegraphics[scale=0.70]{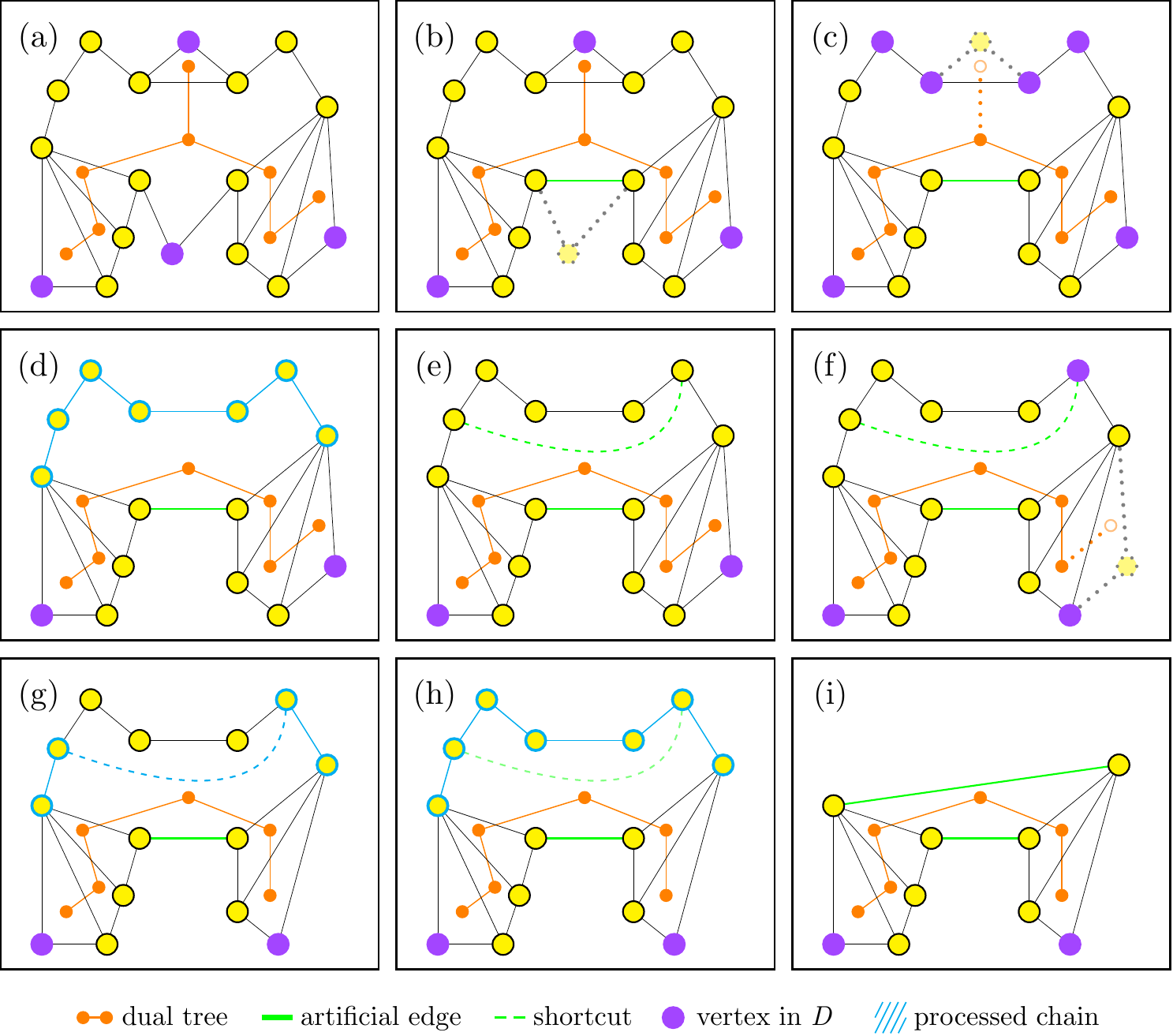}%
 \caption{%
  Snapshots of an execution of Phase 2 are shown.
  The graph after Phase 1 is depicted in (a).
  In (b) the chain
  that contains the bottom middle vertex
  is replaced by an artificial edge.
  In (c) the chain
  that contains the top middle vertex is removed,
  which results in the insertion of four
  new vertices into the choice dictionary.
  In (d) %
  a chain is checked if it is good.
  Since the check failed, %
  a shortcut that connects the vertices of degree~$2$
  near both ends of the chain is created.
  After %
  the removal of the chain
  that contains the right bottom vertex in (f), two vertices
  are inserted
  into the choice dictionary.
  In (g), the upper vertex of the %
  two newly inserted vertices
  is extracted and its chain is determined.
  To check %
  whether the chain (marked in (h)) is good or not, we can use
  the shortcut. 
  Because the chain is good
  it is replaced by an artificial edge in (i).
 }%
 \label{fig:op:p2}%
\end{figure}

\subsection{Correctness of the Algorithm}

Since
the removal of chains in both phases 
is performed in accordance
with Lemma~\ref{lem:bop:prop}
and Lemma~\ref{lem:bop:proof}, to show the correctness of the algorithm,
it remains to verify
that the checks on the counters~$P_{e}$
are correct and sufficient.

\begin{lemma}\label{lem:PeCounter}
 The counter~$P_{e}$
 counts for any edge~$e = \{v_{1},v_{2}\}$
 that belongs at some moment
 during our algorithm
 to the current graph~$G'$
 the number of internal vertex-disjoint paths
 with at least $2$ edges
 between~$v_{1}$ and~$v_{2}$
 that have been removed from~$G$
 so far.
\end{lemma}

\begin{proof}
 Let~$C$ denote a chain
 with endpoints~$u$ and~$v$
 that is to be removed from the graph.
 We now prove
 that this operation
 results in the number
 of internal vertex-disjoint paths
 with at least $2$ edges
 between the endpoints~$v_{1}$
 and~$v_{2}$
 of an edge~$e = \{v_{1}, v_{2}\}$
 to be lowered by~$1$
 for every edge on~$C$
 as well as the edge
 that connects the endpoints of~$C$
 and remains unchanged
 for every other edge.

 It is easy to see
 that the removal of~$C$
 results in the number
 of internal vertex-disjoint paths
 between~$u$ and~$v$
 to be lowered by~$1$.
 Let us now focus on an edge~$e$ on~$C$ and assume that 
 $C$ is split into three parts
 such that~$C = C_{e,1}, e, C_{e,2}$.
 If the original graph
 was biconnected outerplanar,
 there 
 is a path~$P$ between~$u$ and~$v$ that does not use $e$.
 Hence,
 for every edge~$e$ on~$C$
 the path~$C_{e,1}^{-1}, P, C_{e,2}^{-1}$
 connected the endpoints of~$e$
 and has been removed
 by the removal of~$C$ where $C_{e,1}^{-1}$ and $C_{e,2}^{-1}$ denotes
 the paths reverse to $C_{e,1}$ and $C_{e,2}$, respectively.

 For every other edge
 the number of internal vertex-disjoint paths
 remains unchanged
 since each other internal vertex-disjoint path
 that contains a vertex
 that is part of~$C$
 has to contain~$u$ as well as~$v$
 and a path containing~$u$ and~$v$
 remains due to
 the (possibly artificial) edge~$\{u,v\}$.
 Inserting an artificial edge~$\{u,v\}$
 does not increase the number of paths
 for any pair of vertices in~$G'$
 since a path containing~$u$ and~$v$
 existed with the chain $C$ before.
\end{proof}

To obtain the number of internal vertex-disjoint paths
of the original input graph
our algorithm
increments the counter~$P_{e}$
for every edge of the chain~$C$
as well as the edge connecting its endpoints
when~$C$ is the chain
that is currently being removed.
A counter~$P_{e}$
is never decreased.
Since we remove only chains
with at least $2$ edges
whose inner vertices have degree~$2$,
when an edge~$e$
is eventually removed from~$G'$,
so is the last remaining
path with at least $2$ edges
between its endpoints.
As a result there is no need
to store the counters for chains
once they have been removed.
When the removal
of the last remaining chain
results in~$G'$
consisting of a single edge~$\{u,v\}$,
there remains not a single
internal vertex-disjoint path
with at least $2$ edges.
Finally
we check the counter~$P_{\{u,v\}}$
for the last remaining edge~$\{u,v\}$
to make sure
that we check the counter
of every edge
that had been part of the graph~$G'$
at one time
during the execution of our algorithm.

By our counters $P_e$ for all original and artificial edges, 
we count the number
of internal vertex-disjoint paths
between vertices
that have been endpoints of a removed chain.
Thus, our tests are
sufficient by Lemma~\ref{lem:bop:proof}.
By Lemma~\ref{lem:bop:prop}~(ii)
the counting of
the internal vertex-disjoint paths
between the endpoints of other edges
and terminating
if one of these counters
exceeds~$2$
is correct.

\subsection{Space-Efficient Implementation and Space Bounds}
Before any allocation of space is performed
the first check of the algorithm
is to verify
that the number~$m$ of edges
within the input graph~$G$
is at most~$2n-3$,
where~$n$ denotes
the number of vertices in~$G$.
From now on, we assume that~$m = O(n)$.
During Phase~$1$
we use a bit vector of~$n$ bits
that allows us
to mark vertices as \hbox{\textit{tried}}.

The current graph~$G'$
is represented as follows.
We use a choice dictionary
of~O(n) bits,
where $i$ is in the set represented by the choice dictionary
exactly if the vertex~$i$
is still part of~$G'$.
In addition, for each vertex $v$ of initial degree~$deg_{G}(v)$, 
we store $\Theta(deg_{G}(v))$ bits.
Within 
this space
we use a choice dictionary~$C_{v}$
with universe $\{1, \ldots, deg_{G}(v)\}$ as already described
above 
and a counter
that maintains the current degree of~$v$
in $G'$.
It follows that,
ignoring 
artificial edges
and shortcuts,
a representation of the current graph~$G'$
with one choice dictionary for each vertex fits in $O(n)$ bits of working space.

We next want to bound
the number of artificial edges and shortcuts.
We start to bound the number
of chains that induce a face
after Phase~$1$.
For this purpose,
let us define the \emph{dual tree}~$T_{G}$
of a biconnected and outerplanar graph~$G$
as the dual graph of~$G$
minus the vertex
that represents the outer face of~$G$.
Since~$G$ is biconnected and outerplanar,
$T_{G}$ is a tree.
The leaves of~$T_{G}$
correspond to those faces of~$G$
that are induced by chains.
Thus, the removal of a closed chain~$C$
that induces a face~$F$ in~$G'$,
which results in the merging
of~$F$ with the outer face,
corresponds to the removal
of the leaf~$F$ in~$T_{G'}$.
\begin{lemma}
 \label{lem:dual:leafbound}
 If the input graph~$G$
 with~$n$ vertices
 is biconnected outerplanar,
 then the number of chains
 that induce a face
 in the current graph~$G'$
 after the~$t$th round
 of Phase~$1$
 is at most~$n/2^{t}$.
\end{lemma}
\begin{proof}
 It is easy to see
 that the initial number of chains
 that induce a face
 before the first round
 is  
at most~$n$.
 Recall that, in each stage
 of Phase~$1$,
 we consider a vertex~$u$ of degree~$2$
 and test whether it is part
 of a closed chain~$C$
 within our current graph~$G'$.
 If so,
 we remove the chain~$C$ from~$G'$
 and thereafter continue recursively
 on one endpoint~$v$ of~$C$
 if the removal of~$C$
 results in~$v$
 becoming a vertex of degree~$2$.
 Let us analyze the modifications of the algorithm 
 in the dual tree. 
 Whenever we remove a chain,
 this means that we remove 
 a leaf
 and then recursively try 
 if its parent thereby has become a leaf
 and can be removed as well
 until a node of degree at least $2$
 is encountered.

 Vertices of a chain~$C$
 that is incident to a face~$F$
 are marked as \hbox{\textit{tried}}
 if~$F$ could not be merged
 with the outer face
 after a face
 that was incident to~$F$
 has been successfully
 merged with the outer face
 by the removal of the chain
 that induced it.
 The removal of~$F$ fails only
 if~$F$ is not induced by~$C$
 and there remain at least two faces
 in~$T_{G'}$
 that are incident on~$F$.
 We conclude
 that~$F$ had degree at least~$3$
 at the beginning of the current round.
 Since in every round
 leaves are removed
 until a node is encountered
 that had at the beginning of the round
 a degree of at least~$3$,
 the number of leaves
 is at least halved in every round.
\end{proof}

\begin{corollary}
\label{cor:dual:leafbound}
 If the input graph~$G$
 with~$n$ vertices
 is biconnected outerplanar,
 the number of closed chains after Phase~$1$
 in the current graph~$G'$
 is~$O(n/\log n)$.
\end{corollary}

We next bound the number of artificial edges and shortcuts.

\begin{lemma}\label{lem:artifEdges}
The initial number $q$ of chains at the beginning of Round 2
that have vertices in $D$
only doubles during Phase 2, and
the number of artificial edges and shortcuts that are simultaneously in
 use by the algorithm is always at most $2q=O(n/\log n)$.
\end{lemma}

\begin{proof}
Before Phase 2, the algorithm uses neither artificial edges nor
shortcuts.\\%
By the corollary above, 
the number of vertices in the choice dictionary $D$ (each represents a closed
chain) is at most $q=O(n/ \log n)$ at the beginning of Phase 2. One may
consider these chains cutting the Hamilton cycle of $G'$ into $q$ parts. 
Since we add  new vertices into $D$ at the end of Phase 2 only if we were in Case
1 before, a
new chain $C'$ is
always neighbored to an old chain $C$, which is then deleted from $G'$ before
a vertex of $C'$ is
added into $D$. Roughly speaking, the chains at the beginning of Phase 2 are
fires that can spread to new chains on the left and on the right along the Hamilton
cycle, but the fire can never return and ends immediately at the moment when
a new chain starts to burn. It is easy to see that at most $2q$ chains are on
fire.
If we define that a chain that is processed with Case 2 or 3 is still on
fire, then we have to store artificial edges or shortcuts only for chains
that are on fire, and thus, their number is bounded by $2q=O(n/\log n)$.
\end{proof}

We can use a bit vector of $O(n)$ bits
to store for every vertex~$v$
the information
whether an artificial edge or shortcut
exists at~$v$ or not.
As a consequence of the last lemma,
we can store
all artificial edges
as well as all shortcuts in a ragged dictionary and
the total space bound of the algorithm
is~$O(n)$ bits.

\subsection{Time bounds}

It remains to determine the 
running time of our algorithm for recognizing
a biconnected outerplanar graph.
We start with two auxiliary lemmas that
help us to analyze the running time.

\begin{lemma}
 \label{lem:bop:nothree2}
 Let~$u$ denote a vertex of degree~$>2$
 within a biconnected outerplanar graph~$G$.
 Then~$u$ is adjacent to at most
 two vertices of degree~$2$.
\end{lemma}
\begin{proof}
 Aiming at a contradiction, let us assume that~$G$ is a biconnected outerplanar graph
 that contains a vertex~$u$
 that is adjacent
 to at least three vertices $x,y,$ and $z$
 of degree~$2$.
 Since~$G$ is biconnected,
 $G$ has to contain
 a path $P_{a,b}$ between each pair $a,b\in \{x,y,z\}$ with $a\neq b$
 such that $P_{a,b}$ does not contain~$u$, and since $x,y,z$ all have degree 2,
 $P_{a,b}$ also does not contain the remaining vertex in $\{x,y,z\}\setminus \{a,b\}$.
 Moreover, $P_{x,y}$ and $P_{y,z}$ have a common vertex, which is the
 neighbor of $y$ that is not $u$.
 Thus, merging all vertices except $\{u,x,y,z\}$ to one vertex $v$, we
 obtain the $K_{2,3}$
 as minor of~$G$,
 which is a contradiction
 to~$G$ being outerplanar.
\end{proof}

\begin{lemma}
\label{lem:timePerRound}
The total time
that is required
to traverse
the adjacency arrays of vertices
is~$O(n)$ for each round in Phase 1.
\end{lemma}

\begin{proof}
We start to prove
that the total time
that is required
to traverse
the adjacency arrays of vertices
with degree at most $4$
is~$O(n)$.
Since we use a choice dictionary
to bound the time to traverse
the adjacency array of each vertex
to its degree in the current graph~$G'$,
traversing the adjacency array of
a vertex~$v$ of degree at most~$4$
takes constant time.
The adjacency array
of a vertex~$v$ is only traversed
if it is a vertex of degree~$2$
that is tested to be part of a closed chain
or it is adjacent to one
and thus the endpoint
of a (possibly closed) good chain.
Since the number
of vertices of degree~$2$
is bounded by~$n$
and since 
 each vertex of degree $2$ (vertices 
of initial degree larger than $2$ are taken into 
consideration only 
after their degree has dropped to $2$)
is considered at most twice
in every round
to be part of a closed chain,
the total time to traverse
the adjacency arrays of vertices
with degree at most $4$
is~$O(n)$.

It remains to prove
that the time
to traverse the adjacency arrays
of vertices with degree at least $5$
is~$O(n)$ as well.
We allow the adjacency array of every vertex~$v$
of degree at least 5
to be traversed two times
without charge.
Every additional traversal
of the adjacency array of~$v$
has to be paid with~${{\mathrm{deg}}^*}(v)$ coins
that where obtained by a neighbor, but not by own coins.
For this purpose, we give 
every vertex~$v$
of degree at least $5$
at the beginning of each round 
$2 {{\mathrm{deg}}^*}(v)$ coins, 
where ${{\mathrm{deg}}^*}(v)$ denotes the degree of~$v$
at the beginning of the current round.
Since the degree of a vertex is never increased
during the computation of our algorithm,
$deg(v) \le {{\mathrm{deg}}^*}(v)$ holds during the whole round.
At the moment where
the degree of~$v$
is lowered to~$2$,
it gives half of its coins to 
the two vertices of degree at most 3 that are closest before and after
$v$
on the Hamiltonian cycle.

We now consider
the number of traversals
of the adjacency array
of a vertex~$v$ of degree at least 5.
Lemma~\ref{lem:bop:nothree2} states
that each vertex
in a biconnected outerplanar graph
is adjacent to at most two vertices
of degree~$2$.
Let~$C_{1}$ denote a chain
with endpoints~$u$ and~$v$
and~$C_{2}$ denote a chain
with endpoints~$v$ and~$w$.
We only discuss $C_{1}$
since the analysis
with~$C_{2}$ is analogous.
The adjacency array of~$v$
can be traversed once,
when a vertex of degree~$2$
that is located on~$C_{1}$
is drawn from the choice dictionary~$D$
and $deg(v) \le \deg(u)$ holds.
As soon as the degree
of an endpoint of~$C_{1}$ drops to~$4$,
a vertex of degree~$2$ from~$C_{1}$
is reinserted into~$\hbox{\textit{D}}$.
There are two possible cases in which the adjacency array of a vertex $v$
is traversed again
while its degree is bigger than~$2$.
The first case is
that additional removals of chains
result in the degree of~$u$
to be lowered to~$2$
and~$C_{1}$ is now part of a bigger chain containing $u$ and
with $v$ and a vertex~$u'$ as endpoints.
In that case
$u$ gives half of its $2deg^*(u)$ coins to~$v$, i.e., $u$ gives
$deg^*(u)\ge deg(u)\ge deg(v)$ coins to~$v$,
which allows~$v$
to pay for one additional traversal
of its adjacency array
and~$v$ and~$u'$
to keep their own purse of coins
for the time
when their own degree drops to~$2$.
In the second case,
the degree of~$u$ has dropped below $5$, but
is still bigger than~$2$.
Since 
the adjacency array of~$v$
is only traversed
if~$deg(v) \le \deg(u)$,
the analysis for vertices
of degree at most $4$
applies to all future traversals
of the adjacency array of~$v$.

To sum up,
the time
that is required
to traverse the adjacency arrays
of vertices of degree at least $5$
is~$O(4 \sum_{v=1}^{n} {{\mathrm{deg}}^*}(v)) = O(m) = O(n)$
for every round.
\end{proof}

\begin{lemma}\label{lem:biOuter}
 There is an algorithm that, 
 given an $n$-vertex biconnected graph~$G$
 in an adjacency array representation
 with cross pointers, runs in $O(n\log\log n)$ time and uses $O(n)$ bits
 of working space, and determines
 whether~$G$ is biconnected outerplanar.
\end{lemma}

\begin{proof}
It remains to show the running time of our algorithm.
Recall that in outerplanar graphs
the number of edges~$m$
is at most $2n-3$,
where~$n$ denotes the number of vertices.
It follows
that we can assume throughout our algorithm
that~$O(m)=O(n)$.
In both phases of our algorithm
there are two main factors
that have to be considered
when determining the time bound,
the time to traverse chains
or rather the total time used
to find the endpoints of a chains,
and the total time required
to iterate over the adjacency arrays. 

{\bf{Time bound of Phase~$1$.}}
Recall that, given a vertex $v$ of degree 2, 
the algorithm has to identify
all vertices of the
chain $C$ that contains $v$---possibly, the 
algorithm stops earlier due to a vertex marked as tried.
Afterwards, the algorithm has to check if $C$
is closed by iterating through the adjacency array of one endpoint of $C$.

The time used to traverse chains
during each round of the first phase
is~$O(n)$
since we consider
each vertex of degree~$2$
that is part of a chain~$C$
at most twice,
once, when a vertex
that is part of~$C$
is extracted
from the choice dictionary~$\hbox{\textit{D}}$,
and once again,
if the degree
of one of the endpoints of~$C$
is lowered to~$2$
due to the removal of another chain.
Thereafter,
the vertices of degree~$2$
within~$C$
have either been removed
from the graph
or marked as~\hbox{\textit{tried}}.

Lemma~\ref{lem:timePerRound} shows that 
the total time
that is used in each round
to iterate through the adjacency arrays
of vertices is linear in the size of the given graph,
and therefore, the total running time
for Phase~$1$ of our algorithm
is~$O(n\log\log n)$.

{\bf{Time bound of Phase~$2$.}}
Recall that in Phase~$2$
the algorithm repeatedly extracts
a vertex~$u$
of degree~$2$
that is part of a chain~$C$
from the choice dictionary~$\hbox{\textit{D}}$.
To find out if a good chain is closed
the algorithm has to traverse an adjacency array.
Since artificial edges are edges
that incident to the outerface
and since such edges can not be used to close a chain,
the algorithm has to consider only at most $4$ original edges,
which can be done in O(1) time
and, thus, is negligible. %

We next bound the running time for the three cases
of the algorithm.
In the first and second case, 
apart from $O(1)$ vertices (endpoints and their neighbors), we remove the vertices that we consider.
In the 
last case,
where~$C$ is not good,
we remove the vertices of degree~$2$
from the choice dictionary~$\hbox{\textit{D}}$
and insert a shortcut
between vertices of degree~$2$ in~$C$
that are adjacent to the endpoints
of~$C$.
Due to the shortcut, we never have to consider the vertices of $C$ 
again
except for
its endpoints and their neighbors. 
If we neglect the time to use an artificial edges or a shortcut,
we can asymptotically bound
the number of accesses to (artificial and original) edges and the running time
by the sum of\vspace{-1mm}
\begin{itemize}
\setlength{\itemsep}{0pt}
\item the number of removed vertices, which
is $O(n)$, and 
\item the number of chains that we consider, which is also $O(n)$,
since we remove a vertex of degree $2$ from further consideration 
in the algorithm (either
explicitly by deleting it or implicitly by
introducing a shortcut)
whenever we consider a chain.\vspace{-1mm}
\end{itemize}
We next bound the time to access 
artificial edges, which is easy since 
we have $O(n)$ such accesses and each access runs in $O(\log\log n)$ time. 
As a consequence, the total time for the accesses to artificial edges is $O(n\log\log n)$.

At last we consider 
shortcuts.
Each shortcut
is queried only a constant number of times
before being removed or replaced
since each shortcut is only accessed
after the chain~$C$
to which the shortcuts belonged
became good or
was merged to a bigger chain~$C'$.
Thus, either the inner vertices of the chain
including the shortcuts are removed
or the shortcuts are replaced
by shortcuts for~$C'$.
We finally show
that the total number of 
shortcuts is $O(n)$.
Since the shortcuts can be embedded in inner faces
of an outerplanar embedding of the given graph,
the obtained graph is outerplanar, i.e.,
the total number of original edges
plus the number of shortcuts is~$O(n)$.
Thus, the accesses to shortcuts run in $O(n\log\log n)$ total time.
\end{proof}

\subsection{Algorithm for General Outerplanar Graphs}

We next sketch the generalization of our recognition algorithm of
biconnected outerplanar graphs to general outerplanar graphs.
Since a graph is outerplanar exactly if all of its biconnected components
are outerplanar, we
can iterate over the biconnected components of a given graph $G$
using our framework of Section~\ref{sec:bc}---to avoid 
using $\Omega(m)$ time or bits if a non-outerplanar, 
dense graph is given, one should initially check that 
$m \le 2n-3$.
For each biconnected component, 
we first stream all its vertices and build an initially empty choice
dictionary $C_{v}$ with~$\mathrm{deg}_{G}(v)$ keys for each such vertex $v$, then stream the edges of the biconnected
component and fill in the indices of each edge in the choice       
dictionaries of its endpoints. In addition, %
we compute and store for each vertex of
the biconnected component its degree. %

Recall that our subgraph~$G'$
is represented by a choice dictionary
that contains those vertices of~$G$
that are still part of~$G'$
and a choice dictionary~$C_{v}$
for every vertex~$v$
such that $i \in \{1, \ldots, deg_{G}(v)\}$
is present in~$C_{v}$
exactly if the $i$th entry
of the adjacency array of~$v$
contains an edge
that is still present~$G'$.
We initialize these choice dictionaries
with the values that are streamed
from the algorithm
that determines biconnected components
to initialize a representation
of a biconnected subgraph.
The time to test
if the subgraph is biconnected outerplanar
is bounded by the number of vertices and edges
of the subgraph alone. %

Finally note that a check if a outerplanar graph $G=(V,E)$ is maximal outerplanar can
be easily performed by checking if $2|V|-3=|E|$.

\begin{theorem}
 There is an algorithm that,  
 given an $n$-vertex graph~$G$
 in an adjacency array representation
 with cross pointers, runs in time $O(n\log\log n)$ and uses $O(n)$ bits
 of working space, and
 determines
 if~$G$ is (maximal) outerplanar.
\end{theorem}

\phantomsection
\addcontentsline{toc}{chapter}{Bibliography}
\bibliography{main}

\end{document}